\documentclass[journal]{IEEEtran}
\usepackage{amsmath,amssymb,amsfonts,amsthm}
\usepackage{graphicx}
\usepackage{algorithm}
\usepackage{algpseudocode}
\usepackage{booktabs}
\usepackage{multirow}
\usepackage{url}
\usepackage{cite}
\usepackage[colorlinks=true,linkcolor=blue,citecolor=blue,urlcolor=blue]{hyperref}

\newtheorem{theorem}{Theorem}
\newtheorem{proposition}[theorem]{Proposition}
\newtheorem{lemma}[theorem]{Lemma}
\newtheorem{corollary}[theorem]{Corollary}
\theoremstyle{definition}
\newtheorem{definition}{Definition}

\newtheorem{remark}{Remark}

\newcommand{\Q}{\mathbb{Q}}
\newcommand{\Z}{\mathbb{Z}}
\newcommand{\C}{\mathbb{C}}
\newcommand{\Gal}{\mathrm{Gal}}
\newcommand{\inv}{\mathrm{inv}}
\newcommand{\nrd}{\mathrm{nrd}}
\newcommand{\OO}{\mathcal{O}}
\newcommand{\Br}{\mathrm{Br}}
\newcommand{\cA}{\mathcal{A}}
\newcommand{\bX}{\mathbf{X}}
\newcommand{\bH}{\mathbf{H}}
\newcommand{\bY}{\mathbf{Y}}
\newcommand{\bW}{\mathbf{W}}
\newcommand{\bG}{\mathbf{G}}
\newcommand{\bR}{\mathbf{R}}
\newcommand{\bE}{\mathbf{E}}

\newcommand{\bD}{\mathbf{D}}
\newcommand{\bs}{\mathbf{s}}
\newcommand{\zt}{\zeta}

\begin{document}
\title{Non-Vanishing Determinant Space--Time Codes from Abelian Crossed Product Algebras}

\author{Farhad~Barzegar, Dariush~Kiani, and Hassan~Khodaiemehr%
\thanks{F. Barzegar and D. Kiani are with the Department of Mathematics and Computer Science, Amirkabir University of Technology (Tehran Polytechnic), Tehran, Iran (e-mails: fbarzegar@aut.ac.ir, dkiani@aut.ac.ir).}%
\thanks{H. Khodaiemehr is with the School of Engineering, Faculty of Applied Science, The University of British Columbia (Okanagan Campus), Kelowna, BC, Canada (e-mail: hassan.khodaiemehr@ubc.ca).}%
\thanks{Part of this work was presented at the 52nd Annual Iranian Mathematics Conference \cite{Barzegar}.}}

\maketitle

\begin{abstract}
Every explicit non-cyclic crossed product space--time code for more than four antennas has relied on transcendental structure constants, and its coding gain is therefore not bounded away from zero. We explain this phenomenon and show how to avoid it. Over a number field, the commuting-generator (tensor product) construction yields a division algebra only in the cyclic case, and a division crossed product of prime-power degree requires two places with full decomposition group; in particular, no $(\Z/2)^r$, $r\ge3$, crossed product division algebra exists over $\Q(i)$. We then introduce Abelian crossed products with a commutator twist, give the general parametric form of their codewords as a $2\times2$ block matrix of cyclic-algebra codewords, and prove the division property with a centralizer argument and a tame norm-residue symbol. This yields algebraic codes over $\Q(i)$ with non-vanishing determinants and closed-form minimum determinants: an information-lossless $4\times4$ $(\Z/2)^2$ code with $\delta_{\min}=1/2000$, and an infinite family of DMT-optimal $8\times8$ $\Z/4\times\Z/2$ codes indexed by primes $p\equiv5\pmod8$, whose commutator and cyclic constants are roots of unity. The member $p=5$, defined over $\Q(\zeta_{40})$, has minimum determinant exactly $1/5$ before normalization, at least $34.7$ times that of the $8\times8$ perfect code at equal energy, and its rate-2 subcode for $8\times2$ systems has at least $5$ times the minimum determinant of the punctured perfect code. A determinant-spectrum analysis shows that at practical error rates the multiplicity of small determinants, rather than their minimum, governs performance, and it explains all simulated results. Under exact maximum-likelihood decoding the $4\times4$ code gains $0.4$--$0.5$~dB over the perfect code in i.i.d., correlated and imperfect-CSI fading, while the transcendental code falls behind at high SNR; an information-lossless $8\times8$ variant outperforms the $8\times8$ perfect code by up to $0.7$~dB under near-ML decoding, and the minimum-determinant advantage of the non-unitary $8\times8$ code remains asymptotic.
\end{abstract}

\begin{IEEEkeywords}
Space--time block codes, crossed product algebras, division algebras, Brauer group, non-vanishing determinant, perfect codes, sphere decoding.
\end{IEEEkeywords}

\section{Introduction}
\IEEEPARstart{A}{lgebraic} space--time block codes (STBCs) obtained from division algebras are fully diverse by construction. When the structure constants are algebraic integers, or have bounded denominators, the codes also have the non-vanishing determinant (NVD) property, which is sufficient for optimality with respect to the diversity--multiplexing gain tradeoff (DMT) \cite{Elia2006,Tavildar2006}. Cyclic division algebras produced the Golden code \cite{Golden} and the perfect codes \cite{Oggier2006,Elia2007}, which are full-rate, fully diverse and NVD, with uniform average energy and cubic shaping. A perfect code built from a cyclic algebra whose non-norm element is a unit of $\Z[i]$ or $\Z[\zt_3]$, however, exists only for $n\in\{2,3,4,6\}$ transmit antennas \cite{BerhuyOggier2009}. For other values of $n$, Elia \emph{et al.} \cite{Elia2007} use the non-integral non-norm element $\gamma=a/\bar a$, which preserves the NVD property but lowers the guaranteed minimum determinant by the factor $|a|^{2(n-1)}$ (for instance, $5^{7}$ when $n=8$).

Other non-commutative structures have also been used: twisted Laurent series rings, for instance, yield high-rate STBCs with full diversity \cite{KhodaiemehrKiani2015}. Crossed product algebras with an arbitrary Galois group $G$ generalize cyclic algebras and were proposed for STBCs in \cite{Shashidhar2006}. Degree-four crossed products with $G\cong(\Z/2)^2$ and algebraic constants were studied in \cite{BerhuyOggier2007,LuzziOggier2011,OggierVehkalahtiHollanti2010}, mainly as a route to fast-decodable multiple-input double-output (MIDO) codes. For larger degrees, the constructions in \cite{Shashidhar2006}, the biquaternion codes of \cite{UngerMarkin2011} and an earlier version of the present work \cite{Barzegar} all rely on \emph{transcendental} structure constants. Transcendental constants make the division property easy to establish through the theory of generic Abelian crossed products \cite{Amitsur1978,Kuyk1972}, but the determinants of the infinite code then lie in a dense subring of $\C$, and no positive lower bound on the coding gain follows. As observed in \cite[Sec.~III-C]{Oggier2006}, this approach "yields a vanishing minimum determinant when the constellation size increases". The natural question, raised by the reviewers of our earlier submission, is whether a crossed product code can improve on the known cyclic codes at all.

This paper answers that question in three steps. We first identify, by theorems, what prevents algebraic non-cyclic constructions and which Galois groups remain admissible over $\Q(i)$. We then introduce a construction that avoids these obstructions, derive the general parametric form of its codewords and a computable division criterion, and build explicit codes for four and eight antennas. Finally, we evaluate these codes against the perfect codes, both through an exact analysis of their determinant spectra and through simulations.

\subsection{Contributions}\label{sec:contrib}
\begin{enumerate}
\item \emph{Obstructions (Theorems~\ref{thm:commuting} and \ref{thm:local}, Corollary~\ref{cor:Z23}).} Over a number field, the commuting-generator construction $U_i^{n_i}=\gamma_i$, $U_iU_j=U_jU_i$ with algebraic $\gamma_i$ never produces a non-cyclic division algebra. A division crossed product of prime-power degree needs two places with full decomposition group, and $(\Z/2)^r$-crossed product division algebras with $r\ge3$ do not exist over $\Q(i)$. These results explain why earlier non-cyclic constructions needed transcendental constants, and they rule out every algebraic version of the $(\Z/2)^3$ example in \cite{Barzegar}.
\item \emph{A twisted construction with an explicit codeword form (Proposition~\ref{prop:valid}, Theorem~\ref{thm:block}).} For $G=\Z/m\times\Z/2$ we give the general codeword matrix of crossed products with a commutator twist $fe=efu$: a $2\times2$ block matrix whose blocks are cyclic-algebra codewords. Its structure yields an encoder with explicit operation counts (Algorithm~\ref{alg:enc}).
\item \emph{A computable division criterion (Theorem~\ref{thm:centralizer}, Lemma~\ref{lem:tame})} that reduces the division property to the order of an element of a finite field.
\item \emph{Explicit NVD codes with closed-form minimum determinants.} The code $\mathcal{C}_4$ (Section~\ref{sec:C4}) is an information-lossless $(\Z/2)^2$ code over $\Q(i)$ with unit-modulus constants and $\delta_{\min}=1/2000$. The code $\mathcal{C}_8$ (Section~\ref{sec:C8}) is a $\Z/4\times\Z/2$ code over $\Q(\zt_{40})$ with root-of-unity commutator and cyclic constants. Its minimum determinant is $1/5$ before normalization, at least $34.7$ times that of the $8\times8$ perfect code \cite{Elia2007} at equal energy.
\item \emph{A family, puncturing and DMT optimality (Theorem~\ref{thm:family}, Proposition~\ref{prop:punct}, Corollary~\ref{cor:dmt}).} $\mathcal{C}_8$ is the optimal member of an infinite family of NVD, DMT-optimal $8\times8$ codes indexed by primes $p\equiv5\pmod8$. Its rate-2 subcode for $8\times2$ systems has a provably $5$ times larger minimum determinant than the punctured perfect code.
\item \emph{Performance analysis (Section~\ref{sec:spectrum}) and evaluation (Section~\ref{sec:sim}).} An exact low-weight determinant-spectrum and union-bound analysis explains why $\mathcal{C}_4$ outperforms the perfect code at practical error rates despite its smaller minimum determinant, why transcendental codes fail only at high SNR, and why the larger minimum determinant of $\mathcal{C}_8$ is an asymptotic advantage. Exact maximum-likelihood (ML) simulations with visited-node complexity, in i.i.d.\ and spatially correlated Rayleigh fading, confirm these predictions.
\end{enumerate}
We state explicitly what the constructions do \emph{not} achieve. The asymptotic coding gain of $\mathcal{C}_4$ is $0.62$~dB below that of the $4\times4$ perfect code, although $\mathcal{C}_4$ performs better at the simulated error rates. The larger minimum determinant of $\mathcal{C}_8$ is obtained with non-cubic shaping and does not lower the block error rate (BLER) at the simulated operating points; the gain at practical error rates is achieved instead by its unitary variant $\mathcal{C}_8^{\mathrm u}$. Finally, the algebraic structure simplifies encoding and analysis but does not reduce the worst-case complexity of ML decoding.

\subsection{Organization and notation}
Section~\ref{sec:model} introduces the system model and the normalizations used in all comparisons. Section~\ref{sec:prelim} collects the algebraic background, and Section~\ref{sec:obstruction} proves the obstruction results and the division criterion. Section~\ref{sec:construction} presents the twisted construction, the parametric form of its codewords and the encoder. Sections~\ref{sec:C4} and \ref{sec:C8} construct the $4\times4$ and $8\times8$ codes. Section~\ref{sec:spectrum} treats decoding and analyzes the determinant spectra, Section~\ref{sec:sim} reports the simulations, and Section~\ref{conclusion} concludes.

For a field extension $L/K$, $N_{L/K}$ denotes the norm and $\OO_L$ the ring of integers; $\zt_m=e^{2\pi i/m}$. For a central simple $K$-algebra $\cA$, $\nrd$ is the reduced norm and $[\cA]\in\Br(K)$ the Brauer class. For a place $v$ of a number field $K$, $K_v$ is the completion and $\inv_v:\Br(K_v)\to\Q/\Z$ the Hasse invariant. The symbol $\otimes$ denotes both the Kronecker product of matrices and the tensor product of algebras.

\section{System Model and Design Criteria}\label{sec:model}
We consider a quasi-static flat-fading MIMO channel with $n$ transmit antennas, $N_r$ receive antennas and a block length of $T=n$ channel uses,
\begin{equation}\label{eq:channel}
\bY=\sqrt{\rho/n}\,\bH\bX+\bW,
\end{equation}
where $\bH\in\C^{N_r\times n}$ and $\bW\in\C^{N_r\times n}$ have i.i.d.\ $\mathcal{CN}(0,1)$ entries, and $\rho$ is the average SNR per receive antenna. Spatially correlated channels are considered in Section~\ref{sec:sim}. Codewords are normalized to $\mathbb{E}\|\bX\|_F^2=n^2$.

A \emph{linear} code is a set $\{\bX(\bs)=\sum_{k=1}^{n^2}s_k\bX_k:\ \bs\in\mathcal{S}^{n^2}\}$ with dispersion matrices $\bX_k$ and QAM symbols $s_k$; equivalently, $\mathrm{vec}(\bX)=\bG\bs$ for a generator matrix $\bG\in\C^{n^2\times n^2}$. Differences of codewords of a linear code are codewords of the infinite code $\mathcal{C}_\infty=\{\bX(\bs):\bs\in\Z[i]^{n^2}\}$, and the rank and determinant criteria \cite{Tarokh1998} therefore concern
\begin{equation}
\delta_{\min}(\mathcal{C}_\infty)=\inf_{\bs\in\Z[i]^{n^2}\setminus\{0\}}|\det\bX(\bs)|^2,
\end{equation}
where the zero codeword is excluded. The code has the NVD property if $\delta_{\min}(\mathcal{C}_\infty)>0$.

\emph{Fair normalization.} Because $\delta_{\min}$ depends on scaling, codes are compared only after the energy normalization $\mathbb{E}\|\bX\|_F^2=n^2E_s$ for i.i.d.\ zero-mean symbols of energy $E_s$, that is, after scaling the dispersion matrices by $c^{1/2}$ with $c=n^2/\|\bG\|_F^2$. We write
\begin{equation}\label{eq:deltaE}
\delta_E=c^{\,n}\,\delta_{\min}(\mathcal{C}_\infty),
\end{equation}
which is the quantity entering the pairwise error probability bound $P(\bX\to\bX')\le(\rho/4n)^{-nN_r}|\det(\bX-\bX')|^{-2N_r}$. When $\bG$ is unitary, $c=1$ and $\delta_E=\delta_{\min}$; this is the convention of \cite{Oggier2006}, in which the $4\times 4$ perfect code has $\delta_{\min}=1/1125$. Two codes with the same $n$ and rate differ in asymptotic coding gain by $10\log_{10}(\delta_E/\delta_E')/n$~dB.

\begin{definition}[Information losslessness \cite{HassibiHochwald2002,Oggier2006}]\label{def:IL}
A linear code with generator $\bG$ is information-lossless if, for Gaussian inputs $\bs$, the mutual information of the equivalent channel $\bH_{\mathrm{eq}}=(\mathbf{I}_n\otimes\bH)\bG$ equals the MIMO capacity, i.e., $\log\det(\mathbf{I}+\tfrac{\rho}{n}\bH_{\mathrm{eq}}^H\bH_{\mathrm{eq}})=n\log\det(\mathbf{I}+\tfrac{\rho}{n}\bH\bH^H)$ for all $\bH$.
\end{definition}
\begin{proposition}\label{prop:IL}
If $\bG$ is unitary up to a scalar factor, the code is information-lossless.
\end{proposition}
\begin{proof}
If $\bG^H\bG=\mathbf{I}$, Sylvester's identity gives $\det(\mathbf{I}+\tfrac{\rho}{n}\bG^H(\mathbf{I}\otimes\bH^H\bH)\bG)=\det(\mathbf{I}+\tfrac{\rho}{n}\,\mathbf{I}\otimes\bH^H\bH)=\det(\mathbf{I}+\tfrac{\rho}{n}\bH^H\bH)^n$.
\end{proof}
Information losslessness is a property of the generator matrix and does not hold for every crossed product code \cite{Shashidhar2006}; it is therefore established separately for each code below.

\section{Algebraic Background}\label{sec:prelim}
A \emph{$K$-algebra} is a ring $\cA$ that contains $K$ in its center; it is then a $K$-vector space with $K$-bilinear multiplication. It is \emph{central simple} if its center is $K$ and it has no nontrivial two-sided ideals, and a \emph{division algebra} if every nonzero element is invertible. Its degree is $\deg\cA=\sqrt{\dim_K\cA}$. By Wedderburn's theorem, $\cA\cong M_k(D)$ for a division algebra $D$; the \emph{index} of $\cA$ is $\deg D$ and its \emph{exponent} is the order of $[\cA]$ in $\Br(K)$. Consequently, $\cA$ is a division algebra if and only if its index equals its degree.

\begin{definition}[Crossed product \cite{Pierce1982,Draxl1983}]
Let $L/K$ be a Galois extension with group $G$, $|G|=n$, and let $c:G\times G\to L^\times$ be a 2-cocycle. The crossed product $(L/K,G,c)=\bigoplus_{g\in G}u_gL$ has multiplication $\lambda u_g=u_g\,g(\lambda)$ and $u_gu_h=u_{gh}c(g,h)$. It is a central simple $K$-algebra of degree $n$ containing $L$ as a maximal subfield, and it is called \emph{cyclic} (resp.\ \emph{Abelian}) if $G$ is cyclic (resp.\ Abelian).
\end{definition}

The left regular representation of $\cA$ on itself, regarded as a right $L$-vector space with basis $\{u_g\}$, is an injective $K$-algebra homomorphism $\cA\to M_n(L)$, $x\mapsto M_x$. Composed with a complex embedding of $L$, it yields complex matrices with $\det M_x=\nrd(x)\in K$. If $\cA$ is a division algebra, every nonzero $M_x$ is invertible, and codes carved from $\cA$ are fully diverse.

We use the following standard facts on Brauer groups of local and global fields \cite{Pierce1982,Reiner1975,BerhuyOggierBook}.
\begin{enumerate}
\item[(F1)] For a finite place $v$, $\inv_v:\Br(K_v)\to\Q/\Z$ is an isomorphism. The invariant vanishes at complex places and lies in $\tfrac12\Z/\Z$ at real places.
\item[(F2)] (Albert--Brauer--Hasse--Noether.) For $[\cA]\in\Br(K)$, $\inv_v(\cA)=0$ for almost all $v$ and $\sum_v\inv_v(\cA)=0$. The index of $\cA$ equals its exponent and equals the least common multiple of the orders of the $\inv_v(\cA)$.
\item[(F3)] If $L$ splits $\cA$ (for example, if $L$ is a maximal subfield), then $[L_w:K_v]\,\inv_v(\cA)=0$ for $w\mid v$. In particular, the local index at $v$ divides $|G_v|$, where $G_v\le G$ is the decomposition group.
\item[(F4)] For a finite extension $F/K$ and $w\mid v$, $\inv_w(\cA\otimes_KF)=[F_w:K_v]\,\inv_v(\cA)$. If $K\subseteq F=L^H\subseteq L$, the centralizer of $F$ in $(L/K,G,c)$ is the crossed product $(L/F,H,c|_{H\times H})$, which is Brauer equivalent to $\cA\otimes_KF$.
\item[(F5)] At a tamely ramified place the decomposition group is \emph{metacyclic}: the inertia group $I_v$ is cyclic and normal, and $G_v/I_v$ is cyclic. In particular, $G_v$ is generated by two elements.
\item[(F6)] The exponent of a cyclic algebra $(L/F,\sigma,a)$ equals the order of $a$ in $F^\times/N_{L/F}(L^\times)$, over global and local fields alike.
\end{enumerate}

\section{Obstructions and a Division Criterion}\label{sec:obstruction}
This section explains why algebraic structure constants are incompatible with the commuting-generator construction, determines which Galois groups remain admissible over $\Q(i)$, and provides the criterion used later to prove the division property.

\subsection{Commuting generators}
The construction of \cite{Barzegar}, based on generic Abelian crossed products \cite{Amitsur1978,Kuyk1972}, takes linearly disjoint cyclic extensions $L_i/K$ of degrees $n_i$ with generators $\sigma_i$, sets $L=L_1\cdots L_r$, and defines $\cA=L[U_1,\dots,U_r]$ by
\begin{equation}\label{eq:commuting}
U_i\lambda=\sigma_i(\lambda)U_i,\qquad U_i^{n_i}=S_i,\qquad U_iU_j=U_jU_i,
\end{equation}
where the $S_i$ are independent transcendental elements over $L$. The algebra $\cA\cong\bigotimes_i(L_i(S)/K(S),\sigma_i,S_i)$ is then a division algebra \cite{Amitsur1978}.

\begin{theorem}\label{thm:commuting}
Let $K$ be a number field and let $\cA$ be defined by \eqref{eq:commuting} with the $S_i$ replaced by algebraic elements $\gamma_i\in K^\times$. Then $\cA$ is a division algebra only if $n_1,\dots,n_r$ are pairwise coprime, in which case $\Gal(L/K)\cong\prod_i\Z/n_i$ is cyclic and $\cA$ is a cyclic algebra.
\end{theorem}
\begin{proof}
With $\gamma_i\in K^\times$, the subalgebras $\cA_i=\bigoplus_k L_iU_i^k=(L_i/K,\sigma_i,\gamma_i)$ commute elementwise and their dimensions multiply to $\dim_K\cA$; hence $\cA\cong\bigotimes_i\cA_i$. Each $\cA_i$ has degree $n_i$, and its exponent divides $n_i$. The exponent of $\cA$ therefore divides $\mathrm{lcm}(n_1,\dots,n_r)$, and by (F2) the same bound holds for its index. The division property requires $\mathrm{ind}(\cA)=\deg\cA=\prod_in_i$, and $\mathrm{lcm}(n_i)=\prod_i n_i$ holds if and only if the $n_i$ are pairwise coprime. In that case $\prod_i\Z/n_i$ is cyclic by the Chinese remainder theorem, and a crossed product with cyclic group is a cyclic algebra.
\end{proof}

\begin{remark}\label{rem:transc}
Over $K(S)$ with transcendental $S_i$, (F2) fails: index and exponent may differ for generic Abelian crossed products \cite{Amitsur1978}, which is precisely why \eqref{eq:commuting} yields non-cyclic division algebras there. Theorem~\ref{thm:commuting} shows that this freedom disappears once the constants are algebraic. The price of transcendental constants is the loss of the NVD property. The determinants of the transcendental codes lie in $\Z[i][S_i^{\pm1}]$, which is dense in $\C$ when $|S_i|=1$, and no positive lower bound follows from the algebra. Indeed, for $\cA_1=(L_1/K,\sigma_1,S)$ of degree $2$ and $p,q\in\Z[i]$ we have $\nrd(p+qU_1)=p^2-Sq^2$. By the Khintchine-type theorem for approximation by ratios of Gaussian integers \cite{Sullivan1982}, for almost every $S$ on the unit circle there are infinitely many $p/q$ with $|\sqrt S-p/q|<|q|^{-2}(\log|q|)^{-1}$, and hence $\inf|\nrd|=0$. Numerically, the $4\times4$ code of \cite{Barzegar} already has $|\det|^2\approx 8.4\times10^{-7}$ among differences of $4$-QAM codewords, against $0.128$ for the code $\mathcal{C}_4$ of Section~\ref{sec:C4} (Table~\ref{tab:delta}).
\end{remark}

\subsection{Local degrees}
\begin{theorem}\label{thm:local}
Let $K$ be a number field and $\cA=(L/K,G,c)$ a division crossed product of degree $n=p^r$, with $p$ prime. Then there exist two distinct finite places $v_1\ne v_2$ of $K$ with full decomposition groups $G_{v_1}=G_{v_2}=G$, except possibly when $n=2$ and one of the two places is real.
\end{theorem}
\begin{proof}
By (F2), $n=\mathrm{ind}(\cA)=\mathrm{lcm}_v\,\mathrm{ord}(\inv_v\cA)$. Since $n$ is a prime power, some $v_1$ satisfies $\mathrm{ord}(\inv_{v_1}\cA)=n$. Again by (F2), $\sum_{v\ne v_1}\inv_v\cA=-\inv_{v_1}\cA$ has order $n$. If every $v\ne v_1$ had an invariant of order dividing $n/p$, their sum would have the same property; hence some $v_2\ne v_1$ has an invariant of order $n$. By (F3), the local index $n$ divides $|G_{v_j}|\le n$, which forces $G_{v_j}=G$. Complex places carry the invariant $0$, and real places carry invariants of order at most $2$ (F1).
\end{proof}

\begin{corollary}\label{cor:Z23}
If $G$ is not metacyclic, both places in Theorem~\ref{thm:local} lie above $p$. In particular:
\begin{enumerate}
\item there is no $(\Z/2)^r$-crossed product division algebra with $r\ge 3$ over $\Q(i)$, $\Q(\zt_3)$ or $\Q$, each of which has a single place above $2$;
\item a $(\Z/2)^2$-crossed product division algebra over $\Q(i)$ requires two places at which $L/K$ has full local degree $4$; at an odd place this forces ramification index $e=2$ and residue degree $f=2$.
\end{enumerate}
\end{corollary}
\begin{proof}
At a place not above $p$ the extension is tamely ramified, and $G_v$ is metacyclic by (F5). A non-metacyclic $G$ can therefore be a decomposition group only above $p$. The group $(\Z/2)^3$ needs three generators, whereas metacyclic groups are generated by two elements; since $\Q(i)$, $\Q(\zt_3)$ and $\Q$ have a unique dyadic place, the two places of Theorem~\ref{thm:local} cannot exist. For $(\Z/2)^2$ at an odd place, $I_v$ is cyclic, hence of order at most $2$, and $G_v/I_v$ is cyclic; consequently $|I_v|=2$ and $f=2$.
\end{proof}

Corollary~\ref{cor:Z23} excludes every algebraic version of the $(\Z/2)^3$, $n=8$ example of \cite{Barzegar} over $\Q(i)$. It also shows that $\Z/4\times\Z/2$ is admissible: this group is metacyclic and can occur at a tame place with $e=4$ and $f=2$ when the residue field has $q\equiv1\pmod 4$. This observation guides the construction of $\mathcal{C}_8$ in Section~\ref{sec:C8}.

\subsection{A computable division criterion}
\begin{theorem}\label{thm:centralizer}
Let $\cA=(L/K,G,c)$ with $G$ Abelian, let $H\le G$ be a subgroup such that every prime dividing $|G|$ divides $|H|$, and let $F=L^H$. Suppose that there is a finite place $v$ of $K$ with $G_v=G$ at which the centralizer $\mathcal{B}=(L/F,H,c|_{H\times H})$ has $\inv_w(\mathcal{B})$ of order $|H|$, where $w$ is the place of $F$ above $v$. Then $\cA$ is a division algebra.
\end{theorem}
\begin{proof}
Since $G_v=G$, the place $w$ is unique and $[F_w:K_v]=|G|/|H|$. By (F3), $\inv_v\cA=j/|G|$ for some integer $j$, and by (F4), $\inv_w\mathcal{B}=(|G|/|H|)\,j/|G|=j/|H|$. This invariant has order $|H|$ exactly when $\gcd(j,|H|)=1$, which by the hypothesis on the primes implies $\gcd(j,|G|)=1$. Hence $\inv_v\cA$ has order $|G|$, and by (F2) the index of $\cA$ equals $|G|=\deg\cA$.
\end{proof}

\begin{lemma}\label{lem:tame}
Let $L_w/F_w$ be a cyclic, totally and tamely ramified extension of local fields of degree $m$ with residue field $\mathbb{F}_q$ (hence $m\mid q-1$), and let $a\in\OO_{F_w}^\times$. The order of $[(L_w/F_w,\sigma,a)]$ in $\Br(F_w)$ equals the multiplicative order of $\bar a^{(q-1)/m}\in\mathbb{F}_q^\times$.
\end{lemma}
\begin{proof}
By (F6), the order is that of $a$ in $F_w^\times/N(L_w^\times)$. Because $L_w/F_w$ is totally ramified, $v_{F}(N(x))=v_L(x)$, and a unit is a norm if and only if it is the norm of a unit. Tame ramification gives $N(1+\mathfrak{p}_L)=1+\mathfrak{p}_F$, and for a unit $x$ we have $N(x)\equiv\bar x^m$, since inertia acts trivially on the residue field. Therefore $\OO_F^\times/N(\OO_L^\times)\cong\mathbb{F}_q^\times/(\mathbb{F}_q^\times)^m$, a cyclic group of order $m$, in which the class of $a$ has the order of $\bar a^{(q-1)/m}$.
\end{proof}

\section{Abelian Crossed Product Codes with a Commutator Twist}\label{sec:construction}
Theorem~\ref{thm:commuting} shows that the generators of a non-cyclic algebraic crossed product must not commute. We therefore twist the commutation relation by an element $u\in L^\times$ and derive the resulting codewords in closed form.

\subsection{The algebra}
Let $G=\langle\sigma\rangle\times\langle\tau\rangle\cong\Z/m\times\Z/2$ and let $L/K$ be Galois with group $G$. Following \cite{BerhuyOggier2007,LuzziOggier2011} for $m=2$ and the general theory of Abelian crossed products \cite{Amitsur1978}, we consider
\begin{equation}\label{eq:twisted}
\cA=\bigoplus_{\substack{0\le i<m\\ 0\le j<2}} e^if^jL,\quad
\begin{aligned}
&\lambda e=e\,\sigma(\lambda),\ \ \lambda f=f\,\tau(\lambda),\\
&e^m=a,\ \ f^2=b,\ \ fe=ef\,u,
\end{aligned}
\end{equation}
with $a\in L^{\sigma}$, $b\in L^{\tau}$ and $u\in L^\times$.
\begin{proposition}\label{prop:valid}
The relations \eqref{eq:twisted} define an associative crossed product $(L/K,G,c)$ if and only if
\begin{equation}\label{eq:cocycle}
N_\sigma(u):=\prod_{k=0}^{m-1}\sigma^k(u)=\frac{a}{\tau(a)},\qquad u\,\tau(u)=\frac{\sigma(b)}{b}.
\end{equation}
For $u=1$ one recovers the commuting construction, which by Theorem~\ref{thm:commuting} never gives a division algebra over a number field when $m$ is even.
\end{proposition}
\begin{proof}[Proof sketch]
Induction on $k$ gives $fe^k=e^kf\,N_k(u)$ with $N_k(u)=\prod_{t<k}\sigma^t(u)$. On the one hand $fe^m=fa$; on the other hand, $fe^m=e^mf\,N_m(u)=af\,N_m(u)=f\,\tau(a)N_m(u)$, where we used $\lambda f=f\tau(\lambda)$. Hence $a=\tau(a)N_m(u)$, which is the first identity. Comparing $f^2e=be=e\sigma(b)$ with $f(efu)=ef^2\tau(u)u=eb\,u\tau(u)$ yields the second identity. Conversely, \eqref{eq:cocycle} is exactly the 2-cocycle condition for the factor set in \eqref{eq:matrix} below; see \cite{BerhuyOggier2007} for $m=2$. Associativity has, in addition, been verified numerically to machine precision for all codes in this paper.
\end{proof}

\subsection{General parametric form of the codewords}
Order the $L$-basis of $\cA$ as $e^kf^l$, with $k=0,\dots,m-1$ and $l=0,1$. For $x=\sum_{i,j}e^if^jx_{ij}$, the entry of the left regular matrix $M_x$ in row $(i',j')$ and column $(k,l)$ is
\begin{multline}\label{eq:matrix}
[M_x]_{(i',j'),(k,l)}=\sigma^k\tau^l(x_{ij})\cdot\big[\tau^l(N_k(u))\big]^{j}\cdot b^{\lfloor (j+l)/2\rfloor}\\
\cdot\big[\tau^{(j+l)\bmod 2}(a)\big]^{\lfloor (i+k)/m\rfloor},
\end{multline}
where $i\equiv i'-k\pmod m$ and $j\equiv j'-l\pmod 2$. This formula corrects and generalizes the matrix formula of \cite{Barzegar}. Grouping the $m$ "even" coefficients $\mathbf{x}^{(0)}=(x_{00},\dots,x_{m-1,0})$ and the $m$ "odd" coefficients $\mathbf{x}^{(1)}=(x_{01},\dots,x_{m-1,1})$ reveals a transparent structure. For $\gamma\in L^\sigma$ and $\mathbf{y}=(y_0,\dots,y_{m-1})\in L^m$, let
\begin{equation}\label{eq:cyclicblock}
C_\gamma(\mathbf{y})=\begin{bmatrix}
y_0&\gamma\sigma(y_{m-1})&\cdots&\gamma\sigma^{m-1}(y_1)\\
y_1&\sigma(y_0)&\cdots&\gamma\sigma^{m-1}(y_2)\\
\vdots&\vdots&\ddots&\vdots\\
y_{m-1}&\sigma(y_{m-2})&\cdots&\sigma^{m-1}(y_0)
\end{bmatrix},
\end{equation}
that is, $[C_\gamma(\mathbf{y})]_{i',k}=\gamma^{[i'<k]}\sigma^k(y_{i'-k})$ with indices modulo $m$. This is the familiar codeword of the cyclic algebra $(L/L^\sigma,\sigma,\gamma)$ \cite{Oggier2006,OggierFnT}. Let $\mathbf{D}_u=\mathrm{diag}(N_0(u),N_1(u),\dots,N_{m-1}(u))$ and write $\tau(\cdot)$ for the entrywise action of $\tau$.

\begin{theorem}[Parametric codeword form]\label{thm:block}
For every $x\in\cA$,
\begin{equation}\label{eq:block}
M_x=\begin{bmatrix}
C_a\big(\mathbf{x}^{(0)}\big) & b\,C_a\big(\tau(\mathbf{x}^{(1)})\big)\,\tau(\mathbf{D}_u)\\[2pt]
C_{\tau(a)}\big(\mathbf{x}^{(1)}\big)\,\mathbf{D}_u & C_{\tau(a)}\big(\tau(\mathbf{x}^{(0)})\big)
\end{bmatrix}.
\end{equation}
If $\{\omega_1,\dots,\omega_n\}$ is a $\Z[i]$-basis of a lattice $\Lambda\subseteq\OO_L$, the codewords are the matrices $\bX(\bs)=c^{1/2}M_x$ with
\begin{equation}\label{eq:layers}
x_{ij}=\sum_{t=1}^{n}s_{ij,t}\,\omega_t,\qquad s_{ij,t}\in\mathcal{S},\quad 0\le i<m,\ j\in\{0,1\},
\end{equation}
which carry $n^2=(2m)^2$ QAM symbols; the code is full-rate.
\end{theorem}
\begin{proof}
Split \eqref{eq:matrix} into the four blocks given by $j'\in\{0,1\}$ and $l\in\{0,1\}$. For $j'=l=0$ we have $j=0$, and the only constant is $a^{\lfloor(i+k)/m\rfloor}=a^{[i'<k]}$, which gives $C_a(\mathbf{x}^{(0)})$. For $j'=1$, $l=0$ we have $j=1$, the factor $N_k(u)$ depends only on the column, and the cyclic constant is $\tau(a)$, which gives $C_{\tau(a)}(\mathbf{x}^{(1)})\mathbf{D}_u$. For $j'=0$, $l=1$ we have $j=1$, the entries are $\tau$-conjugated, the factor $\tau(N_k(u))$ depends only on the column, $b$ appears once and the cyclic constant is $a$, which gives the upper-right block. For $j'=l=1$ we have $j=0$ and obtain $C_{\tau(a)}(\tau(\mathbf{x}^{(0)}))$.
\end{proof}

\begin{remark}
For $m=1$ (that is, $G=\Z/2$), \eqref{eq:block} reduces to the quaternion-type codeword $\left[\begin{smallmatrix}x_0&b\,\tau(x_1)\\x_1&\tau(x_0)\end{smallmatrix}\right]$, of which the Alamouti and Golden codes are special cases. Theorem~\ref{thm:block} thus shows that a twisted $\Z/m\times\Z/2$ crossed product code is a quaternion-like arrangement of four cyclic codewords of size $m$, coupled by the commutator $u$ through $\mathbf{D}_u$ and by the constant $b$. For $u=1$ and $a\in K$ the blocks decouple into the Kronecker structure of \cite{Barzegar}, which is never a division algebra over a number field (Theorem~\ref{thm:commuting}).
\end{remark}

Equivalently, embedding $L\subset\C$ entrywise,
\begin{equation}\label{eq:codeword}
\bX=c^{1/2}\sum_{i,j}\bE_e^{\,i}\bE_f^{\,j}\,\bD(x_{ij}),\qquad \bD(\lambda)=\mathrm{diag}\big(g(\lambda)\big)_{g\in G},
\end{equation}
where $\bE_e=M_e$ and $\bE_f=M_f$ are \emph{monomial} matrices (permutation times diagonal) whose nonzero entries are products of conjugates of $a$, $b$ and $u$. In particular, every entry of $\bX$ equals exactly one entry of one layer $\bD(x_{ij})$ multiplied by a fixed constant.

\subsection{Shaping and determinants}
\begin{proposition}[Unitary shaping]\label{prop:unitary}
Suppose that the embedding matrix $\bR=[g(\omega_t)]_{g,t}$ of the basis in \eqref{eq:layers} is unitary up to scale and that all conjugates of $a$, $b$ and $u$ have modulus $1$. Then $\bG$ is unitary up to scale, and the code is information-lossless.
\end{proposition}
\begin{proof}
By \eqref{eq:codeword}, $\mathrm{vec}(\bX)=c^{1/2}\mathbf{P}\boldsymbol{\Phi}(\mathbf{I}_n\otimes\bR)\bs$, where $\mathbf{P}$ is a permutation matrix and $\boldsymbol\Phi$ is diagonal with unit-modulus entries. Proposition~\ref{prop:IL} applies.
\end{proof}

\begin{proposition}[NVD with explicit bound]\label{prop:nvd}
Assume that $a,u\in\OO_L^\times$, $\Lambda\subseteq\OO_L$, and $b=\beta/\pi$ with $\beta\in\OO_L$ and $\pi\in\OO_K$. Let $\mathfrak{p}$ be a prime of $K$ and $e$ the ramification index of the primes of $L$ above it. If $v_{\mathfrak{P}}(b)\ge -1$ for every prime $\mathfrak{P}\mid\mathfrak p$ and $b$ is integral at all other primes, then $\nrd(x)\in\mathfrak{p}^{-\lfloor n/(2e)\rfloor}$ for all $x\in\bigoplus e^if^j\Lambda$. Consequently, $|\nrd(x)|^2\ge N(\mathfrak{p})^{-\lfloor n/(2e)\rfloor}$ for $x\ne0$ when $K=\Q(i)$.
\end{proposition}
\begin{proof}
By Theorem~\ref{thm:block}, $b$ occurs only in the upper-right $\tfrac n2\times\tfrac n2$ block of $M_x$, once in each entry. Every term of the Leibniz expansion of $\det M_x$ uses at most $n/2$ entries of that block, and therefore $v_{\mathfrak P}(\det M_x)\ge-n/2$. Since $\det M_x\in K$ and $v_{\mathfrak P}=e\,v_{\mathfrak p}$ on $K$, we obtain $v_{\mathfrak p}(\nrd x)\ge -n/(2e)$, and hence $v_{\mathfrak p}(\nrd x)\ge-\lfloor n/(2e)\rfloor$ because the valuation is an integer. At all other primes every entry is integral.
\end{proof}

\subsection{Encoding}
By \eqref{eq:codeword}, encoding requires only the $n$ layer vectors $\bR\bs_{ij}$ and a fixed placement of their entries. When $\bR=\bR_1\otimes\cdots\otimes\bR_r$ with $\bR_t\in\C^{n_t\times n_t}$, as for all codes in this paper (their layer bases are products of bases of cyclic subfields), the layer transform costs $n\sum_t n_t$ complex multiplications instead of $n^2$.

\begin{algorithm}[t]
\caption{Encoder for twisted Abelian crossed product codes}\label{alg:enc}
\begin{algorithmic}[1]
\Require symbols $\bs\in\mathcal{S}^{n^2}$ split into layers $\bs_{ij}\in\mathcal{S}^n$; layer factors $\bR_1,\dots,\bR_r$; precomputed placement $\pi_{ij}:G\to G$ and unit constants $\phi_{ij}(g)$ from \eqref{eq:block}
\Ensure codeword $\bX\in\C^{n\times n}$
\State $\bX\gets\mathbf{0}_{n\times n}$
\For{each layer $(i,j)$}
  \State $\mathbf{v}\gets(\bR_1\otimes\cdots\otimes\bR_r)\bs_{ij}$ \Comment{$n\sum_tn_t$ mults}
  \For{each $g\in G$} \Comment{column index}
     \State $\bX[\pi_{ij}(g),g]\gets\phi_{ij}(g)\,\mathbf{v}[g]$ \Comment{one mult per entry}
  \EndFor
\EndFor
\State \Return $c^{1/2}\bX$ \Comment{energy normalization}
\end{algorithmic}
\end{algorithm}

\begin{proposition}\label{prop:enccost}
Algorithm~\ref{alg:enc} computes a codeword with $n^2(1+\sum_tn_t)$ complex multiplications, compared with $n^4$ for a generic linear-dispersion code (dense $\bG$) and $n^3+n^2$ for a cyclic code with a dense layer matrix. For $n=8$ and $\bR=\bR_4\otimes\bR_2$ the three counts are $448$, $4096$ and $576$.
\end{proposition}
Each entry of a codeword depends on exactly one layer; this is the precise sense in which the tensor structure of \cite{Barzegar} simplifies encoding. The structure does \emph{not} reduce the worst-case complexity of ML decoding, which remains $|\mathcal{S}|^{n^2}$ for exhaustive search, as for all full-rate $n\times n$ codes without fast-decodable structure \cite{Biglieri2009,JithamithraRajan}.

\section{The $4\times4$ Code $\mathcal{C}_4$}\label{sec:C4}
We first apply the construction with $m=2$. The field is chosen according to Corollary~\ref{cor:Z23}(2), which requires two odd primes with $e=f=2$.

\subsection{Construction and explicit codewords}
Let $K=\Q(i)$ and $L=K(\sqrt5,\zt_8)=\Q(i,\sqrt2,\sqrt5)$, the field of the simulation example in \cite{Barzegar}. Let $\sigma$ fix $\sqrt5$ with $\sigma(\zt_8)=-\zt_8$, and let $\tau$ fix $\zt_8$ with $\tau(\sqrt5)=-\sqrt5$. We take
\begin{equation}\label{eq:C4par}
a=\frac{\sqrt5}{2+i},\qquad b=\zt_8,\qquad u=i.
\end{equation}
Then $a\in L^\sigma$, $b\in L^\tau$, $N_\sigma(u)=u^2=-1=a/\tau(a)$ and $u\tau(u)=-1=\sigma(b)/b$; hence \eqref{eq:cocycle} holds. All three constants have modulus one under every embedding.

For the layers we use the ideal $\alpha\OO_L$ with $\theta=(1+\sqrt5)/2$ and $\alpha=1+i-i\theta$, which satisfies $\alpha\bar\alpha=2+i$ (the Golden code ideal \cite{Golden}). Its basis $\{\alpha,\alpha\theta\}\otimes\{1,\zt_8\}$ has embedding matrix $\bR=\bR_5\otimes\bR_8$ with
\[
\bR_5=\tfrac{1}{\sqrt5}\begin{bmatrix}\alpha&\alpha\theta\\ \bar\alpha&\bar\alpha\bar\theta\end{bmatrix},\qquad
\bR_8=\tfrac{1}{\sqrt2}\begin{bmatrix}1&\zt_8\\1&-\zt_8\end{bmatrix},
\]
both unitary, where $\bar\alpha=\tau(\alpha)=1+i-i\bar\theta$ and $\bar\theta=(1-\sqrt5)/2$. By Proposition~\ref{prop:unitary}, $\mathcal{C}_4$ is information-lossless and has cubic shaping. Writing $x_0,x_1,x_2,x_3$ for the coefficients of $1,e,f,ef$, Theorem~\ref{thm:block} gives the codewords
\begin{equation}\label{eq:C4cw}
\bX=\frac{1}{\sqrt{10}}\begin{bmatrix}
x_0 & a\,\sigma(x_1) & \zt_8\,\tau(x_2) & i a\zt_8\,\sigma\tau(x_3)\\
x_1 & \sigma(x_0) & \zt_8\,\tau(x_3) & i\zt_8\,\sigma\tau(x_2)\\
x_2 & -ia\,\sigma(x_3) & \tau(x_0) & -a\,\sigma\tau(x_1)\\
x_3 & i\,\sigma(x_2) & \tau(x_1) & \sigma\tau(x_0)
\end{bmatrix},
\end{equation}
with
\begin{equation}\label{eq:C4layers}
x_j=\alpha\big(s_{j,1}+s_{j,2}\theta+s_{j,3}\zt_8+s_{j,4}\theta\zt_8\big),\qquad s_{j,t}\in\mathcal{S},
\end{equation}
and with $\sigma(\zt_8)=-\zt_8$, $\tau(\theta)=\bar\theta$, $\tau(\alpha)=\bar\alpha$. Each codeword carries $16$ QAM symbols.

\subsection{Division property and minimum determinant}
\begin{theorem}\label{thm:C4div}
$\cA_4$ is a division algebra of degree $4$, and its local invariants at $2\pm i$ have order $4$.
\end{theorem}
\begin{proof}
We verify the two places required by Corollary~\ref{cor:Z23}. Since $5=(2+i)(2-i)$, $\sqrt5$ ramifies at both primes ($e=2$). The element $\zt_8$ generates an extension that is inert at both, because $\zt_8^2=i\equiv 3$ (resp.\ $2$) modulo $2+i$ (resp.\ $2-i$), and these residues are non-squares in $\mathbb{F}_5$. Hence $G_v=G$ for $v=(2\pm i)$. We apply Theorem~\ref{thm:centralizer} with $H=\langle\sigma\rangle$ and $F=L^\sigma=K(\sqrt5)$. The centralizer $\mathcal{B}=(L/F,\sigma,a)$ is the quaternion algebra $(i,a)_F$, because $L=F(\sqrt i)$. At the prime $\mathfrak{p}$ of $F$ above $2+i$, the element $i$ is a unit and $v_{\mathfrak p}(a)=v_{\mathfrak p}(\sqrt5)-v_{\mathfrak p}(2+i)=1-2=-1$ is odd. The tame Hilbert symbol is therefore the Legendre symbol of $i\equiv3$ in $\mathbb{F}_5$, which equals $-1$. Consequently, $\inv_{\mathfrak p}\mathcal{B}=\tfrac12$ has order $|H|=2$, and $\cA_4$ is a division algebra. The same computation at $2-i$, where $i\equiv2$, shows that the invariant there also has order $4$, as Theorem~\ref{thm:local} requires. This agrees with the criterion of \cite{BerhuyOggier2007}: the quaternion algebra $(-5,i)_{\Q(i)}$ ramifies exactly at $2\pm i$.
\end{proof}

\begin{theorem}\label{thm:C4det}
$\delta_{\min}(\mathcal{C}_{4,\infty})=1/2000$.
\end{theorem}
\begin{proof}
In \eqref{eq:C4cw}, $\bX=M_x/\sqrt{10}$ with $x_j\in\alpha\OO_L$. The prime $(\alpha)$ of $K(\sqrt5)$ above $2+i$ is inert in $L$; we denote it by $\mathfrak{P}$, with $e(\mathfrak P|2+i)=2$. Every entry of $M_x$ lies in $\mathfrak P$ times a constant. By \eqref{eq:C4cw}, the constant $a$, with $v_{\mathfrak P}(a)=-1$, appears in at most two entries of any Leibniz term (it occurs only in column~$2$ and in column~$4$ of rows $1$ and $3$), while $b$ and $u$ are units. Each Leibniz term therefore has $v_{\mathfrak P}\ge 4-2=2$, which gives $(2+i)\mid\nrd(x)$. At the prime above $2-i$ the constant $a$ has positive valuation, and at all other primes every entry is integral. Hence $\nrd(x)\in(2+i)\Z[i]\setminus\{0\}$, $|\nrd(x)|^2\ge5$, and $|\det\bX|^2=|\nrd(x)|^2/10^4\ge1/2000$. Equality is attained by $x=-\alpha-i\alpha\theta\zt_8+e\,(-i\alpha\theta\zt_8)$, for which $|\nrd(x)|^2=5$.
\end{proof}

\begin{remark}
The $4\times4$ perfect code has $\delta_{\min}=1/1125$ \cite{Oggier2006}; the asymptotic coding gain of $\mathcal{C}_4$ is therefore $10\log_{10}(2000/1125)/4=0.62$~dB lower. The loss stems from the denominator of $a$. Over $\Q$, Remark~1 of \cite{LuzziOggier2011} shows that unit-modulus constants are incompatible with the division property in degree four. Over $\Q(i)$, every unit-modulus choice of $(a,b,u)$ in this family that we examined (all associates of $\sqrt5/(2\pm i)$, $\zt_8$ and $\pm i$) gives exactly $1/2000$. We include $\mathcal{C}_4$ as a fully specified algebraic Abelian non-cyclic $4\times4$ code over $\Q(i)$ with unit-modulus constants, NVD, cubic shaping and information losslessness, and as a validation of the theory. Section~\ref{sec:spectrum} shows that its determinant spectrum is nevertheless more favorable than that of the perfect code.
\end{remark}

\section{The $8\times8$ Codes}\label{sec:C8}
For eight antennas, Corollary~\ref{cor:Z23} excludes $(\Z/2)^3$ and leaves $\Z/4\times\Z/2$, which we realize with $m=4$. Unlike the cyclic case, the construction allows the commutator and cyclic constants to be roots of unity.

\subsection{The code $\mathcal{C}_8$ and its explicit codewords}\label{sec:C8def}
Let $K=\Q(i)$ and $L=\Q(\zt_{40})=K(\zt_5,\zt_8)$, with $G=\Gal(L/K)=\langle\sigma\rangle\times\langle\tau\rangle$. Here $\sigma:\zt_{40}\mapsto\zt_{40}^{17}$ has order $4$, with $\sigma(\zt_5)=\zt_5^2$ and $\sigma(\zt_8)=\zt_8$, and $\tau:\zt_{40}\mapsto\zt_{40}^{21}$ has order $2$, with $\tau(\zt_5)=\zt_5$ and $\tau(\zt_8)=-\zt_8$. Let
\begin{equation}
\beta_0=\zt_5+i\zt_5^2-\zt_5^4-i\zt_5^3,\qquad \sigma(\beta_0)=-i\beta_0,
\end{equation}
a quartic Gauss sum with $|\beta_0|^2=5$ and $\beta_0^4=-15+20i=-(2-i)^3(2+i)$. We set
\begin{equation}\label{eq:C8par}
u=\zt_8,\qquad a=\zt_8^{-1},\qquad b=\frac{\beta_0}{2-i}.
\end{equation}
Then $a\in L^\sigma=\Q(\zt_8)$, $b\in L^\tau=K(\zt_5)$, $N_\sigma(u)=\zt_8^4=-1=a/\tau(a)$ and $u\tau(u)=-i=\sigma(b)/b$; hence \eqref{eq:cocycle} holds. The constants $u$ and $a$ are roots of unity, and $|b|=1$ under every embedding.

Since $\sigma$ fixes $\zt_8$, $N_k(u)=\zt_8^k$, and Theorem~\ref{thm:block} with $\tau(a)=-\zt_8^{-1}$ gives the codewords
\begin{equation}\label{eq:C8cw}
\bX=\frac{1}{2\sqrt2}\begin{bmatrix}
C_{\zt_8^{-1}}\big(\mathbf{x}^{(0)}\big) & b\,C_{\zt_8^{-1}}\big(\tau(\mathbf{x}^{(1)})\big)\,\bar{\boldsymbol\Lambda}\\[2pt]
C_{-\zt_8^{-1}}\big(\mathbf{x}^{(1)}\big)\,\boldsymbol\Lambda & C_{-\zt_8^{-1}}\big(\tau(\mathbf{x}^{(0)})\big)
\end{bmatrix},
\end{equation}
where $\boldsymbol\Lambda=\mathrm{diag}(1,\zt_8,i,\zt_8^3)$, $\bar{\boldsymbol\Lambda}=\tau(\boldsymbol\Lambda)=\mathrm{diag}(1,-\zt_8,i,-\zt_8^3)$, and $C_\gamma(\cdot)$ is the $4\times4$ cyclic block \eqref{eq:cyclicblock} with $\sigma(\zt_5)=\zt_5^2$. The eight layers are
\begin{equation}\label{eq:C8layers}
x_{ij}=\sum_{t=0}^{3}\big(s_{ij,t,0}+s_{ij,t,1}\zt_8\big)\zt_5^{\,t},\qquad s_{ij,t,l}\in\mathcal{S},
\end{equation}
with $\mathbf{x}^{(0)}=(x_{00},x_{10},x_{20},x_{30})$ the coefficients of $1,e,e^2,e^3$ and $\mathbf{x}^{(1)}=(x_{01},\dots,x_{31})$ those of $f,ef,e^2f,e^3f$. Each codeword carries $64$ QAM symbols, and the factor $1/(2\sqrt2)=c^{1/2}$ normalizes the energy (Theorem~\ref{thm:C8det}).

\subsection{Division property and minimum determinant}
\begin{theorem}\label{thm:C8div}
$\cA_8$ is a division algebra of degree $8$.
\end{theorem}
\begin{proof}
At $v=(2+i)$, $\zt_5$ generates a totally and tamely ramified extension of degree $4$, and $\zt_8$ generates an unramified extension of degree $2$ ($5$ has order $2$ in $(\Z/8)^\times$). Hence $G_v=G$, and $\Z/4\times\Z/2$ occurs with $e=4$ and $f=2$, as permitted by Corollary~\ref{cor:Z23}. We take $H=\langle\sigma\rangle$, for which $|G|=8$ and $|H|=4$ have the same prime divisors, and $F=L^\sigma=\Q(\zt_8)$. The centralizer is the cyclic algebra $\mathcal{B}=(L/F,\sigma,a)$, and $L_w/F_w$ is totally and tamely ramified of degree $4$ with residue field $\mathbb{F}_{25}$. Reduction is injective on roots of unity of order prime to $5$; consequently $\bar a=\bar\zt_8^{-1}$ has order $8$ in $\mathbb{F}_{25}^\times$, and $\bar a^{(25-1)/4}=\bar\zt_8^{-6}=\bar\zt_8^{\,2}$ has order $4$. By Lemma~\ref{lem:tame}, $\inv_w\mathcal{B}$ has order $4=|H|$, and Theorem~\ref{thm:centralizer} completes the proof.
\end{proof}

Theorem~\ref{thm:C8div} does not depend on $b$. In contrast with cyclic codes, the division property is carried by the root-of-unity constant $a$, and $b$ can be chosen freely to optimize the determinant.

\begin{theorem}\label{thm:C8det}
With the basis \eqref{eq:C8layers}, the infinite code satisfies $\min_{x\ne0}|\det M_x|^2=1/5$ exactly. After energy normalization ($c=1/8$), $\delta_E(\mathcal{C}_8)=1/(5\cdot2^{24})$.
\end{theorem}
\begin{proof}
From $\beta_0^4=-(2-i)^3(2+i)$ and $\beta_0\bar\beta_0=5$, the prime $\mathfrak{P}'$ of $L$ above $2-i$, with $e=4$, satisfies $v_{\mathfrak P'}(\beta_0)=3$, and hence $v_{\mathfrak P'}(b)=3-4=-1$; at the prime above $2+i$, $v(b)=1$. Proposition~\ref{prop:nvd} with $n=8$ and $e=4$ gives $\nrd(x)\in(2-i)^{-1}\Z[i]$, and therefore $|\nrd(x)|^2\ge1/5$. The value $1/5$ is attained by the two-symbol element $x=-\zt_8+f\,(i\zt_5^3)$. All dispersion matrices have entries of modulus $1$, which gives $\|\bG\|_F^2=8\cdot 64$, $c=64/512=1/8$, and the second claim follows from \eqref{eq:deltaE}.
\end{proof}

\emph{Comparison with the perfect code.} The $8\times8$ perfect code \cite{Elia2007} uses $L=\Q(\zt_{32})$, $\gamma=(2+i)/(1+2i)$ and the unitary generator $\tfrac1{\sqrt8}[\zt_{32}^{\,j5^k}]$. Its guaranteed bound is $\delta_E\ge1/(5^7\cdot2^{24})$ in the normalization of Theorem~\ref{thm:C8det}, which is $5^6=15625$ times smaller than the value for $\mathcal{C}_8$. To obtain a conservative comparison we also searched the infinite perfect code; the smallest value found is $\delta_E\approx3.43\times10^{-10}=0.00576/2^{24}$, which is an upper bound on its true minimum. Hence
\begin{equation}
\frac{\delta_E(\mathcal{C}_8)}{\delta_E(\text{perfect }8\times8)}\ \ge\ \frac{1/5}{0.00576}\ =\ 34.7,
\end{equation}
that is, at least $10\log_{10}(34.7)/8=1.93$~dB more asymptotic coding gain. We note that the original $8\times8$ example of \cite{Barzegar} (cyclic, $L=\Q(\zt_{32})$, transcendental $S$) is the perfect code with $\gamma$ replaced by $S$.

\emph{Shaping and the unitary variant $\mathcal{C}_8^{\mathrm u}$.} The basis \eqref{eq:C8layers} is not orthogonal: the Gram matrix of $\{\zt_5^t\}$ under $\mathrm{Tr}_{L/K}(x\bar y)$ is $5\mathbf{I}-\mathbf{J}$, and $\bG$ has condition number $\sqrt5$. The corresponding real lattice is of $A_4$ type and denser than $\Z^8$; this is a relaxed-shaping trade-off in the spirit of \cite{SrinathRajanImproved}, at the cost of information losslessness. The loss is quantified in Fig.~\ref{fig:mi}: about $0.9$ bit per channel use at $10$~dB and $1.6$ bits at $20$~dB, i.e., $3.6\%$ of the $44$-bit capacity. For applications that require information losslessness, the orthogonal normal basis generated by $w=1+(1+i)\zt_5^3$, namely $\{\sigma^t(w)\zt_8^l\}$, gives a unitary variant $\mathcal{C}_8^{\mathrm u}$ (Proposition~\ref{prop:unitary}) with the same codeword form \eqref{eq:C8cw} and $\delta_E\approx1/(5.12\times10^{9})$ (numerical), $0.3$~dB below the perfect code.

\subsection{An infinite family generalizing $\mathcal{C}_8$}
The construction of $\mathcal{C}_8$ uses only three properties of the prime $5$: $5\equiv1\pmod4$, which makes $5$ split in $\Z[i]$; $5\equiv5\pmod8$, which makes $\zt_8$ inert above $5$; and the existence of a cyclic quartic subfield of $\Q(\zt_5)$.

\begin{theorem}\label{thm:family}
Let $p\equiv5\pmod 8$ be prime, $p=\pi\bar\pi$ in $\Z[i]$, $M\subset\Q(\zt_p)$ the cyclic quartic subfield, $K=\Q(i)$ and $L=K M(\zt_8)$. Let $g$ be a primitive root modulo $p$, $\sigma:\zt_p\mapsto\zt_p^g$ (of order $4$ on $KM$), $\tau:\zt_8\mapsto-\zt_8$, and let $\chi$ be the quartic character with $\chi(g)=i$ and $G(\chi)=\sum_{x}\chi(x)\zt_p^x$ its Gauss sum. Put
\[
u=\zt_8,\qquad a=\zt_8^{-1},\qquad b=G(\chi)/\pi_0,
\]
where $\pi_0\in\{\pi,\bar\pi\}$ is the prime with $\pi_0^3\mid G(\chi)^4$. Then:
\begin{enumerate}
\item \eqref{eq:cocycle} holds, $|a|=|b|=|u|=1$, and the crossed product $\cA_8^{(p)}$ is a division algebra of degree $8$;
\item with the integral basis (Gaussian periods of $M$)$\otimes\{1,\zt_8\}$ in every layer, $|\nrd(x)|^2\ge1/p$ for all nonzero $x$; the code is NVD, and its codewords have the form \eqref{eq:C8cw} with $b$ as above.
\end{enumerate}
For $p=5$ this is $\mathcal{C}_8$.
\end{theorem}
\begin{proof}
(1) $\sigma(G(\chi))=\sum_x\chi(x)\zt_p^{gx}=\overline{\chi(g)}\,G(\chi)=-iG(\chi)$, and therefore $\sigma(b)/b=-i=u\tau(u)$; moreover $N_\sigma(u)=\zt_8^4=-1=a/\tau(a)$ and $|G(\chi)|^2=p=|\pi_0|^2$. At $v=(\pi)$, the extension $KM/K$ is totally and tamely ramified of degree $4$, and $\zt_8$ generates the residue extension of degree $2$, because $p\equiv5\pmod8$ implies that the Frobenius of $p$ acts as $\zt_8\mapsto\zt_8^5=-\zt_8$. Hence $G_v=G$. The residue field of $F=\Q(\zt_8)$ at $w\mid v$ is $\mathbb{F}_{p^2}$. Since $p\equiv5\pmod 8$, $v_2(p^2-1)=3$, and $\bar a^{(p^2-1)/4}=\bar\zt_8^{-(p^2-1)/4}$ has order $8/\gcd(8,(p^2-1)/4)=4$. Lemma~\ref{lem:tame} and Theorem~\ref{thm:centralizer} with $H=\langle\sigma\rangle$ give the division property.
(2) Classically, $G(\chi)^2=J(\chi,\chi)G(\chi^2)$ and $G(\chi^2)^2=p$ \cite[Ch.~8]{IrelandRosen}; therefore $G(\chi)^4=pJ(\chi,\chi)^2$ with $J(\chi,\chi)\in\Z[i]$ and $|J|^2=p$. Thus $J$ is an associate of $\pi$ or $\bar\pi$, and $G(\chi)^4$ is an associate of $\pi_0^3\bar\pi_0$, where $\pi_0$ is the prime dividing $J$. The primes of $L$ above $\pi_0$ and $\bar\pi_0$ have ramification index $4$; hence $v_{\mathfrak{P}_0}(G(\chi))=3$, $v_{\mathfrak{P}_0}(b)=-1$, and $v(b)=+1$ above $\bar\pi_0$. Because $\mathrm{disc}(K)$ and $\mathrm{disc}(M)$ are coprime, $\OO_{KM}=\Z[i]\otimes\OO_M$, and the Gaussian periods form a $\Z$-basis of $\OO_M$. Proposition~\ref{prop:nvd} with $n=8$ and $e=4$ gives $\nrd(x)\in\pi_0^{-1}\Z[i]$.
\end{proof}

Table~\ref{tab:family} shows that the energy-normalized minimum determinant decreases rapidly with $p$, because the Gaussian periods grow in norm and the layer lattice becomes less orthogonal (condition number $\sqrt p$). The prime $p=5$ is therefore the optimal member of the family, which justifies the choice of $\mathcal{C}_8$.

\begin{table}[t]
\centering
\caption{The family of Theorem~\ref{thm:family} (raw $=\min|\nrd|^2$ by search; bound $=1/p$, proved)}\label{tab:family}
\begin{tabular}{@{}cccccc@{}}
\toprule
$p$ & $\pi_0$ & bound & raw & cond$(\bG)$ & $\delta_E$\\
\midrule
5 & $2-i$ & $1/5$ & $1/5$ & $2.24$ & $1/8.39\cdot10^{7}$\\
13 & $3-2i$ & $1/13$ & $1$ & $3.61$ & $1/2.56\cdot10^{10}$\\
29 & $5+2i$ & $1/29$ & $1$ & $5.39$ & $1/1.41\cdot10^{13}$\\
\bottomrule
\end{tabular}
\end{table}

\subsection{Rate-$2$ puncturing for $8\times2$ systems}
With $N_r=2$ receive antennas, full-rate $8\times8$ codes cannot be decoded efficiently, and rate-$N_r$ codes are used instead \cite{VHO2012,SrinathRajan2014}. Keeping the two layers $\{1,e\}$ of \eqref{eq:C8cw}, i.e., setting $\mathbf{x}^{(1)}=0$ and $x_{20}=x_{30}=0$, gives $16$ QAM symbols per $8$ channel uses. The codewords become
\begin{equation}\label{eq:C8x2cw}
\bX=\frac{1}{2\sqrt2}\begin{bmatrix}
C_{\zt_8^{-1}}(\mathbf{y}) & \mathbf{0}\\
\mathbf{0} & C_{-\zt_8^{-1}}\big(\tau(\mathbf{y})\big)
\end{bmatrix},
\end{equation}
where $\mathbf{y}=(x_{00},x_{10},0,0)$. The codeword is a block-diagonal matrix of two $4\times4$ cyclic codewords linked by $\tau$.

\begin{proposition}\label{prop:punct}
Let $\mathcal{C}_8^{\times2}$ denote the code \eqref{eq:C8x2cw} with layers \eqref{eq:C8layers}. Then $\min_{x\ne0}|\det M_x|^2=1$ exactly, and $\delta_E(\mathcal{C}_8^{\times2})=1/2^{24}$ under the normalization $\mathbb{E}\|\bX\|_F^2=16E_s$. The subcode of the $8\times8$ perfect code on the layers $\{1,e\}$ has $\delta_E\le1/(5\cdot2^{24})$ (search); the ratio is at least $5$, i.e., $10\log_{10}5/8=0.87$~dB more asymptotic coding gain.
\end{proposition}
\begin{proof}
In \eqref{eq:C8x2cw} the constant $b$ does not appear, and the remaining constants $\pm\zt_8^{-1}$ are roots of unity. Consequently $M_x$ has entries in $\OO_L$ and $\nrd(x)\in\Z[i]\setminus\{0\}$, which gives $|\nrd(x)|^2\ge1$; equality holds for $x=1$. The normalization follows as in Theorem~\ref{thm:C8det}, since all entries of the dispersion matrices have modulus one.
\end{proof}
The same argument applies to every subcode supported on the even layers $1,e,e^2,e^3$ (that is, with $\mathbf{x}^{(1)}=0$), since $b$ multiplies only the odd layers in \eqref{eq:block}. This makes $\mathcal{C}_8$ a natural mother code for rate-$1$ to rate-$4$ asymmetric codes with a common decoder structure. Being a subset of a division algebra, every such subcode is fully diverse.

\subsection{DMT optimality}
\begin{corollary}\label{cor:dmt}
The codes $\mathcal{C}_4$, $\mathcal{C}_8$, $\mathcal{C}_8^{\mathrm u}$ and every member of the family of Theorem~\ref{thm:family} achieve the optimal diversity--multiplexing tradeoff of the $n\times N_r$ channel for every $N_r\ge1$.
\end{corollary}
\begin{proof}
Each of these codes is a full-rate $n\times n$ code ($n^2$ QAM symbols per $n$ channel uses) with the NVD property (Theorems~\ref{thm:C4det}, \ref{thm:C8det} and \ref{thm:family}, and Proposition~\ref{prop:nvd} for $\mathcal{C}_8^{\mathrm u}$). By \cite{Elia2006}, full-rate NVD codes achieve the DMT \cite{ZhengTse} for every number of receive antennas. The original transcendental codes of \cite{Barzegar,Shashidhar2006} do not satisfy this hypothesis.
\end{proof}

\section{Decoding and Determinant-Spectrum Analysis}\label{sec:spectrum}
This section describes the ML decoder used for all codes and analyzes the codeword differences that dominate the error probability; the analysis predicts the simulation results of Section~\ref{sec:sim}.

\subsection{Maximum-likelihood decoding}\label{sec:dec}
Stacking columns, $\mathrm{vec}(\bY)=\sqrt{\rho/n}(\mathbf{I}_n\otimes\bH)\bG\bs+\mathrm{vec}(\bW)$. The equivalent $2nN_r\times 2n^2$ real model $\mathbf{y}=\mathbf{B}\mathbf{z}+\mathbf{w}$, with $\mathbf{z}=[\Re\bs^T,\Im\bs^T]^T$, is decoded by the Schnorr--Euchner sphere decoder \cite{Agrell2002} with sorted-QR ordering \cite{Wubben2001}, which is exact ML. For $n=8$ with $N_r=8$ (128 real dimensions), exact ML decoding is impractical at low and moderate SNR, and the K-best decoder with $K=1024$ is applied to the same model for all codes; Section~\ref{sec:sim} shows that the ranking of the codes does not depend on $K$. Complexity is reported as the average number of visited nodes.

\begin{remark}[Vectorization]
Any other stacking (rows instead of columns, or interleaved real and imaginary parts) is a fixed permutation of the coordinates of $\mathbf{y}$ and $\mathbf{z}$, i.e., an orthogonal change of variables; the ML decision and the error performance are therefore unchanged. Sphere-decoding complexity depends on the ordering through the zero pattern of the $\mathbf R$-factor. With sorted QR the decoder re-orders the columns for every channel realization, which makes the visited-node counts essentially insensitive to the initial vectorization. Orderings that expose Hurwitz--Radon orthogonality are those that enable fast decodability \cite{Biglieri2009,JithamithraRajan}.
\end{remark}

\subsection{Determinant spectrum and union bound}
The minimum determinant governs the asymptote of slope $nN_r$ of the error probability. At practical error rates, however, the union bound is dominated by the \emph{number} of codeword differences with small eigenvalue products. For a linear code with $4$-QAM symbols, let $\Delta=\bX(\boldsymbol\delta)$ be the difference matrix for a vector $\boldsymbol\delta$ with entries in $\{0,\pm2,\pm2i,\pm2\pm2i\}$, and let $w(\boldsymbol\delta)=\prod_k 2^{-\#\{\text{nonzero real parts of }\delta_k\}}$ be the probability that a uniformly chosen codeword admits $\boldsymbol\delta$ as a difference. The Chernoff bound for Rayleigh fading gives
\begin{equation}\label{eq:ub}
P_{\mathrm{blk}}\le\sum_{\boldsymbol\delta\ne0}w(\boldsymbol\delta)\prod_{i=1}^{n}\Big(1+\frac{\rho\,\lambda_i(\Delta\Delta^H)}{4n}\Big)^{-N_r}.
\end{equation}
We evaluated \eqref{eq:ub} exactly over all differences of Hamming weight at most $3$ for the $4\times4$ and $8\times2$ codes ($2.9\times10^5$ terms) and at most $2$ for the $8\times8$ codes ($1.3\times10^5$ terms). Table~\ref{tab:spectrum} reports, in the normalization of Section~\ref{sec:sim}, the smallest determinant $d_1=\min|\det\Delta|^2$ and the weighted multiplicities $A_1$ (differences attaining $d_1$) and $A_4$ (differences with $|\det\Delta|^2\le4d_1$).

\begin{table}[t]
\centering
\caption{Low-weight determinant spectrum, 4-QAM, normalized codewords}\label{tab:spectrum}
\begin{tabular}{@{}lccc@{}}
\toprule
Code & $d_1$ & $A_1$ & $A_4$\\
\midrule
Perfect $4\times4$ & $1.42\cdot10^{-2}$ & $492$ & $1580$\\
$\mathcal{C}_4$ & $8.00\cdot10^{-3}$ & $304$ & $496$\\
Original $4\times4$ & $5.28\cdot10^{-8}$ & $16$ & $16$\\
\midrule
Perfect $8\times8$ & $8.79\cdot10^{-8}$ & $192$ & $384$\\
$\mathcal{C}_8$ & $3.05\cdot10^{-6}$ & $576$ & $1600$\\
$\mathcal{C}_8^{\mathrm u}$ & $1.00\cdot10^{-7}$ & $1024$ & $1792$\\
Original $8\times8$ & $4.63\cdot10^{-9}$ & $64$ & $64$\\
\midrule
Punct.\ perfect $8\times2$ & $0.200$ & $512$ & $640$\\
$\mathcal{C}_8^{\times2}$ & $1.000$ & $320$ & $2496$\\
\bottomrule
\end{tabular}
\end{table}

\begin{figure}[t]
\centering
\includegraphics[width=\columnwidth]{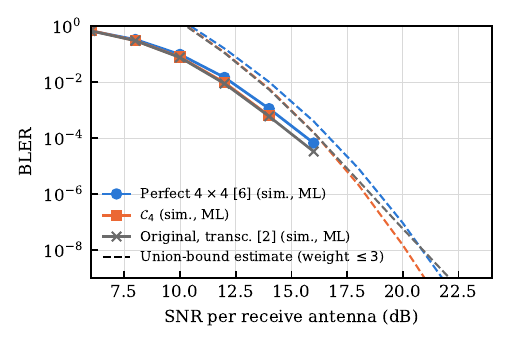}
\caption{Union-bound estimate \eqref{eq:ub} over differences of Hamming weight $\le3$ (dashed) and simulated exact-ML BLER (solid), $4\times4$, $4$-QAM, $N_r=4$. The estimate reproduces the ordering of the codes and predicts that the transcendental code falls behind at high SNR.}\label{fig:ub4}
\end{figure}

Table~\ref{tab:spectrum} and Fig.~\ref{fig:ub4} lead to three predictions, all confirmed in Section~\ref{sec:sim}.
\begin{itemize}
\item $\mathcal{C}_4$ has a $1.8$ times smaller $d_1$ than the perfect code but $3.2$ times fewer near-minimal differences ($A_4=496$ against $1580$). It should therefore outperform the perfect code at moderate error rates despite its $0.62$~dB asymptotic deficit.
\item The original transcendental codes have extremely small $d_1$ but very small multiplicity ($A_4=16$). They should be indistinguishable from the algebraic codes at moderate SNR and fall behind only at high SNR.
\item $\mathcal{C}_8$ and $\mathcal{C}_8^{\times2}$ have larger $d_1$, by exactly the proved factors $34.7$ and $5$, but larger $A_4$. Under \eqref{eq:ub}, the curves of $\mathcal{C}_8$ and the perfect code cross only near $26$~dB, at a BLER of about $10^{-37}$ ($N_r=8$), and the $8\times2$ curves cross near $36$~dB, at about $10^{-34}$. In the practical range the multiplicity term dominates, and the larger minimum determinant acts as an asymptotic guarantee rather than a practical gain.
\end{itemize}
The design lesson is that for $n\ge8$, the relevant figure of merit at practical error rates is the low-determinant multiplicity $A_4$ rather than $d_1$. Optimizing $A_4$ within the algebraic families constructed here, for instance through the choice of ideal, layer basis and $b$, is a well-posed problem that we leave open.

\section{Numerical Results}\label{sec:sim}
All codes are implemented from their published generator matrices. As a check, our determinant routine recovers exactly the published minimum determinants of the Golden code ($1/5$) and of the $4\times4$ perfect code ($1/1125$), and the explicit forms \eqref{eq:C4cw} and \eqref{eq:C8cw} coincide with the dispersion-matrix implementation to machine precision. The sphere decoder was verified against brute-force ML, and the GPU-batched K-best decoder against a sequential implementation. Exact-ML points use up to $3\times10^6$ codewords ($4\times4$ and $8\times2$) and K-best points up to $5\times10^5$ codewords; each point stops after $200$--$300$ block errors, and a point with no error in $3\times10^6$ codewords ends the curve. Three channel models are used: i.i.d.\ Rayleigh fading; Kronecker-correlated fading with exponential correlation $0.7$ at both ends; and imperfect channel knowledge, where the receiver uses $\hat\bH=\bH+\mathbf{E}$ with i.i.d.\ $\mathcal{CN}(0,1/\rho)$ estimation errors.

\subsection{Minimum determinants and information loss}
Table~\ref{tab:delta} collects the energy-normalized minimum determinants \eqref{eq:deltaE}. Entries marked "proved" are exact; entries marked "search" are the smallest values found by random search over sparse integer symbol vectors and are therefore upper bounds on the true minimum. The last column gives the minimum of $|\det(\bX-\bX')|^2$ over pairs of $4$-QAM codewords, the quantity that governs the pairwise error probability of the transmitted (finite) code.

\begin{table}[t]
\centering
\caption{Energy-normalized minimum determinants}\label{tab:delta}
\setlength{\tabcolsep}{3pt}
\begin{tabular}{@{}llccc@{}}
\toprule
$n$ & Code & $\delta_E$ (infinite code) & IL & 4-QAM \\
\midrule
\multirow{3}{*}{4} & Perfect \cite{Oggier2006} & $1/1125$ (proved) & yes & $0.228$\\
 & $\mathcal{C}_4$ (new) & $1/2000$ (proved) & yes & $0.128$\\
 & Original \cite{Barzegar} & no positive bound & yes & $8.4\cdot10^{-7}$\\
\midrule
\multirow{4}{*}{8} & Perfect \cite{Elia2007} & $\ge 1/1.3\cdot10^{12}$ (proved) & yes & --\\
 & & $\le 1/2.91\cdot10^{9}$ (search) & & \\
 & $\mathcal{C}_8$ (new) & $1/8.39\cdot10^{7}$ (proved) & no & --\\
 & $\mathcal{C}_8^{\mathrm u}$ (new) & $\ge 1/1.3\cdot10^{11}$ (proved) & yes & --\\
 & & $\le 1/5.12\cdot10^{9}$ (search) & & \\
 & Original $\Z/8$ \cite{Barzegar} & $\le 1/4.7\cdot10^{11}$ (search) & yes & --\\
\bottomrule
\end{tabular}
\vspace{2pt}

{\footnotesize IL: information-lossless. For $\mathcal{C}_8^{\mathrm u}$, the lower bound follows from Proposition~\ref{prop:nvd}, since its layer lattice lies in $\OO_L$.}
\end{table}

The table supports three observations. First, the earlier transcendental codes have determinants several orders of magnitude below those of every algebraic code, even within the $4$-QAM codebook. Second, among the $8\times8$ codes, the proved determinant of $\mathcal{C}_8$ exceeds the best value found for the perfect code by a factor of $34.7$. Third, the unitary variant $\mathcal{C}_8^{\mathrm u}$ lies between the two.

\begin{figure}[t]
\centering
\includegraphics[width=\columnwidth]{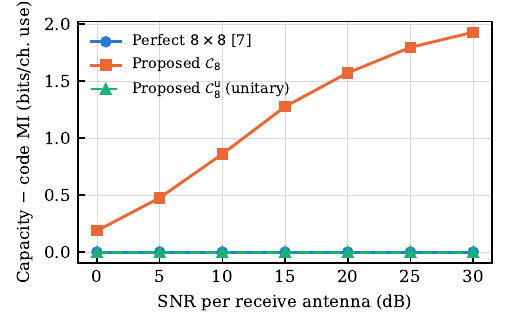}
\caption{Gap between the $8\times8$ MIMO capacity and the mutual information of the code with Gaussian inputs, i.i.d.\ Rayleigh fading, $N_r=8$ (300 channel realizations per point). The perfect code and $\mathcal{C}_8^{\mathrm u}$ are information-lossless (zero gap); $\mathcal{C}_8$ loses $0.86$ bit at $10$~dB and $1.57$ bits at $20$~dB.}\label{fig:mi}
\end{figure}

\subsection{Error performance of the $4\times4$ codes (exact ML)}
Fig.~\ref{fig:bler4} shows the BLER of the $4\times4$ codes with $4$-QAM (32 bits per codeword) under exact ML decoding, down to BLER $10^{-6}$; Table~\ref{tab:gaps4} lists the SNR differences.
\begin{itemize}
\item In i.i.d.\ fading, $\mathcal{C}_4$ outperforms the perfect code by $0.18$~dB at BLER $10^{-2}$ and by $0.45$--$0.48$~dB between $10^{-3}$ and $10^{-5}$, despite its smaller asymptotic $\delta_{\min}$. This is the behaviour predicted by its smaller multiplicity $A_4$ (Section~\ref{sec:spectrum}). The original transcendental code performs close to $\mathcal{C}_4$ in this range.
\item In correlated fading, $\mathcal{C}_4$ outperforms the perfect code by $0.38$--$0.41$~dB between $10^{-3}$ and $10^{-5}$. The transcendental code, in contrast, falls \emph{behind} the perfect code by $0.22$~dB at $10^{-3}$, $0.40$~dB at $10^{-4}$ and $0.43$~dB at $10^{-5}$; at $20$~dB its BLER is $3.6$ times that of $\mathcal{C}_4$. The growing gap is the simulated manifestation of its small determinants (Table~\ref{tab:spectrum} and Remark~\ref{rem:transc}).
\item With imperfect channel knowledge, $\mathcal{C}_4$ outperforms the perfect code by $0.24$~dB at $10^{-2}$ and by $0.31$--$0.43$~dB between $10^{-3}$ and $10^{-4}$ (Fig.~\ref{fig:bler4csi}).
\end{itemize}

\begin{table}[t]
\centering
\caption{SNR difference (dB) to the $4\times4$ perfect code at target BLER, exact ML, $4$-QAM (negative: better)}\label{tab:gaps4}
\begin{tabular}{@{}lcccc@{}}
\toprule
 & $10^{-2}$ & $10^{-3}$ & $10^{-4}$ & $10^{-5}$\\
\midrule
\multicolumn{5}{@{}l}{i.i.d.\ Rayleigh}\\
\quad $\mathcal{C}_4$ & $-0.18$ & $-0.46$ & $-0.48$ & $-0.45$\\
\quad Original & -- & $-0.37$ & $-0.44$ & --\\
\multicolumn{5}{@{}l}{Correlated ($\rho_t=\rho_r=0.7$)}\\
\quad $\mathcal{C}_4$ & -- & $-0.39$ & $-0.41$ & $-0.38$\\
\quad Original & -- & $+0.22$ & $+0.40$ & $+0.43$\\
\multicolumn{5}{@{}l}{Imperfect CSI ($\sigma_e^2=1/\rho$)}\\
\quad $\mathcal{C}_4$ & $-0.24$ & $-0.31$ & $-0.43$ & --\\
\quad Original & $-0.22$ & $-0.33$ & $-0.40$ & --\\
\bottomrule
\end{tabular}
\vspace{2pt}

{\footnotesize --: target not bracketed by the simulated points.}
\end{table}

\begin{figure}[t]
\centering
\includegraphics[width=\columnwidth]{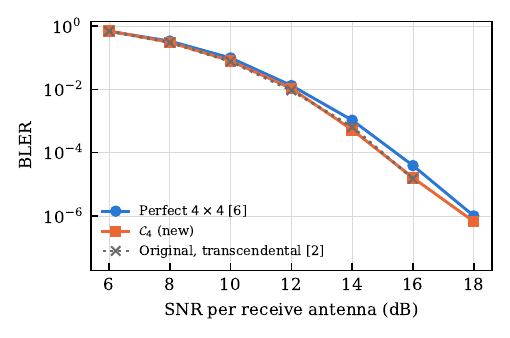}\\
\includegraphics[width=\columnwidth]{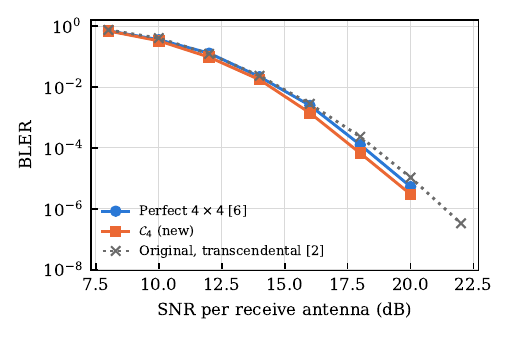}
\caption{BLER of $4\times4$ codes with $4$-QAM and exact ML (sphere) decoding, $N_r=4$. Top: i.i.d.\ Rayleigh. Bottom: Kronecker-correlated Rayleigh, exponential correlation $0.7$ at both ends.}\label{fig:bler4}
\end{figure}

\begin{figure}[t]
\centering
\includegraphics[width=\columnwidth]{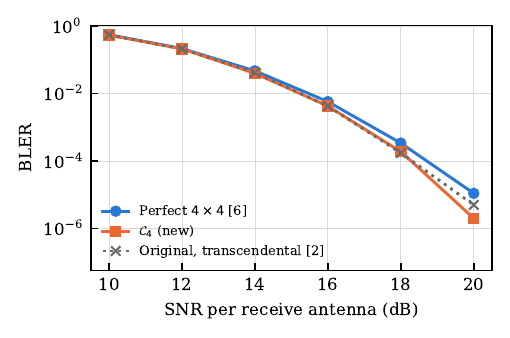}
\caption{BLER of $4\times4$ codes with $4$-QAM and exact ML decoding under imperfect channel knowledge ($\sigma_e^2=1/\rho$), $N_r=4$.}\label{fig:bler4csi}
\end{figure}

\emph{Larger constellations.} Fig.~\ref{fig:bler4hq} shows $16$-QAM and $64$-QAM results with K-best decoding ($K=2048$ and $K=4096$). The three codes lie within $0.3$~dB of one another down to BLER $10^{-2}$, and all curves flatten below about $10^{-3}$. The flattening is a decoder floor rather than a property of the codes: it occurs at the same error rate for all three codes and is absent from the exact-ML $4$-QAM curves. An attempt at node-capped exact ML decoding for $16$-QAM did not resolve this regime, because the node cap was reached in a fraction of codewords comparable to the BLER itself; those runs are therefore not reported. Within the simulable range, the vanishing determinant of the transcendental code does not become visible for these constellation sizes.

\begin{figure}[t]
\centering
\includegraphics[width=\columnwidth]{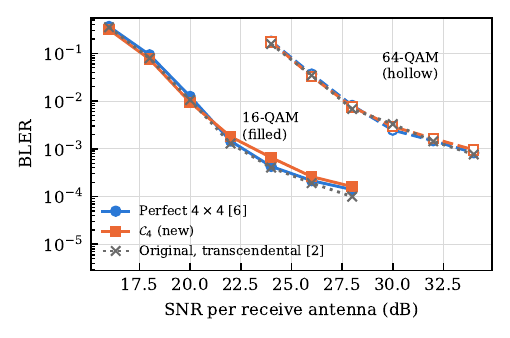}
\caption{BLER of $4\times4$ codes with $16$-QAM (filled markers, K-best with $K=2048$) and $64$-QAM (hollow markers, $K=4096$), i.i.d.\ Rayleigh, $N_r=4$.}\label{fig:bler4hq}
\end{figure}

\subsection{Decoding complexity}
We measured the average number of visited nodes of the sphere decoder (sorted QR, $4$-QAM). In i.i.d.\ fading the three codes are within a factor of $1.5$ of one another (at $12$~dB: $4.0\cdot10^3$ for the perfect code, $4.3\cdot10^3$ for $\mathcal{C}_4$ and $2.9\cdot10^3$ for the transcendental code). In correlated fading, $\mathcal{C}_4$ visits about $35\%$ more nodes than the perfect code ($5.0\cdot10^4$ against $3.7\cdot10^4$ at $12$~dB). In agreement with Section~\ref{sec:construction}, the algebraic structure does not reduce ML decoding complexity; its benefit is the simpler encoder of Proposition~\ref{prop:enccost}.

\subsection{Error performance of the $8\times8$ codes (near ML)}
Exact ML decoding of the 128-dimensional real model is out of reach in this SNR range; node-capped sphere decoding hit its cap in a fraction of codewords comparable to the BLER and is not reported. All $8\times8$ codes are therefore decoded by the same GPU-batched K-best decoder. Table~\ref{tab:ksweep} shows that increasing $K$ from $64$ to $4096$ lowers the BLER of every code by more than an order of magnitude, which confirms that the decoder still contributes to the error rate, but the \emph{ranking} $\mathcal{C}_8^{\mathrm u}\preceq$ perfect $\prec\mathcal{C}_8$ is the same for every $K$. The comparisons below use $K=1024$.

\begin{table}[t]
\centering
\caption{BLER of $8\times8$ codes versus the K-best list size $K$ (i.i.d.\ Rayleigh, $4$-QAM, $N_r=8$)}\label{tab:ksweep}
\setlength{\tabcolsep}{4pt}
\begin{tabular}{@{}lcccc@{}}
\toprule
 & $K=64$ & $256$ & $1024$ & $4096$\\
\midrule
\multicolumn{5}{@{}l}{$16$~dB}\\
\quad Perfect & $4.2\cdot10^{-2}$ & $1.7\cdot10^{-2}$ & $7.7\cdot10^{-3}$ & $1.9\cdot10^{-3}$\\
\quad $\mathcal{C}_8$ & $6.9\cdot10^{-2}$ & $3.0\cdot10^{-2}$ & $1.2\cdot10^{-2}$ & $3.7\cdot10^{-3}$\\
\quad $\mathcal{C}_8^{\mathrm u}$ & $3.3\cdot10^{-2}$ & $1.6\cdot10^{-2}$ & $5.3\cdot10^{-3}$ & $1.6\cdot10^{-3}$\\
\multicolumn{5}{@{}l}{$20$~dB}\\
\quad Perfect & $9.8\cdot10^{-3}$ & $3.9\cdot10^{-3}$ & $1.5\cdot10^{-3}$ & $2.6\cdot10^{-4}$\\
\quad $\mathcal{C}_8$ & $1.7\cdot10^{-2}$ & $6.8\cdot10^{-3}$ & $2.5\cdot10^{-3}$ & $3.6\cdot10^{-4}$\\
\quad $\mathcal{C}_8^{\mathrm u}$ & $8.6\cdot10^{-3}$ & $3.5\cdot10^{-3}$ & $1.1\cdot10^{-3}$ & $1.6\cdot10^{-4}$\\
\bottomrule
\end{tabular}
\end{table}

\emph{i.i.d.\ fading} (Fig.~\ref{fig:bler8}). The unitary code $\mathcal{C}_8^{\mathrm u}$ outperforms the perfect code by $0.53$~dB at BLER $10^{-2}$ and by $0.70$~dB at $10^{-3}$; the two curves meet near $5\cdot10^{-4}$ ($24$~dB: $2.7\cdot10^{-4}$ against $2.5\cdot10^{-4}$, with $54$ and $49$ errors). Both codes are information-lossless, and the gain therefore stems from the determinant spectrum of $\mathcal{C}_8^{\mathrm u}$ rather than from shaping. The code $\mathcal{C}_8$ is $1.06$~dB worse than the perfect code at $10^{-2}$ and $1.26$~dB worse at $10^{-3}$, with no crossing through $24$~dB. As predicted in Section~\ref{sec:spectrum}, its larger minimum determinant is an asymptotic guarantee that does not translate into a gain at practical error rates; the non-orthogonal generator (Fig.~\ref{fig:mi}) and the higher multiplicity of small determinants work against it. The original transcendental code tracks the perfect code down to $10^{-3}$ and falls $0.38$~dB behind at $5\cdot10^{-4}$.

\emph{Correlated fading and imperfect CSI} (Fig.~\ref{fig:bler8b}). With correlation $0.7$, $\mathcal{C}_8^{\mathrm u}$ outperforms the perfect code by $0.74$~dB at BLER $10^{-2}$, the transcendental code is within $0.15$~dB of the perfect code, and $\mathcal{C}_8$ is $1.4$~dB worse at $10^{-1}$. With imperfect channel knowledge, $\mathcal{C}_8^{\mathrm u}$ gains $0.34$~dB at $10^{-2}$ and $0.71$~dB at $10^{-3}$, and $\mathcal{C}_8$ is $1.2$~dB worse at $10^{-2}$.

\begin{figure}[t]
\centering
\includegraphics[width=\columnwidth]{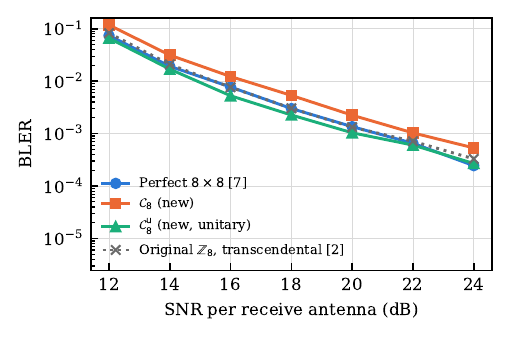}
\caption{BLER of $8\times8$ codes with $4$-QAM, i.i.d.\ Rayleigh, $N_r=8$, K-best decoding ($K=1024$) with sorted QR.}\label{fig:bler8}
\end{figure}

\begin{figure}[t]
\centering
\includegraphics[width=\columnwidth]{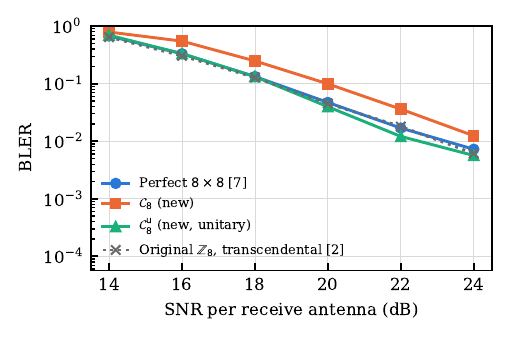}\\
\includegraphics[width=\columnwidth]{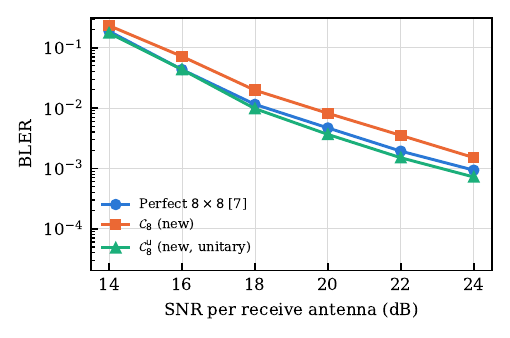}
\caption{BLER of $8\times8$ codes with $4$-QAM, $N_r=8$, K-best ($K=1024$). Top: Kronecker-correlated Rayleigh, correlation $0.7$. Bottom: imperfect channel knowledge, $\sigma_e^2=1/\rho$.}\label{fig:bler8b}
\end{figure}

\subsection{Error performance of the $8\times2$ codes (exact ML)}
Fig.~\ref{fig:bler82} shows the rate-2 subcodes of Proposition~\ref{prop:punct} with $N_r=2$, decoded by exact ML (32 real dimensions) down to BLER $8\cdot10^{-6}$.
\begin{itemize}
\item $\mathcal{C}_8^{\times2}$ remains $0.18$--$0.23$~dB behind the punctured perfect code between BLER $10^{-2}$ and $10^{-4}$; at $16$~dB the two codes reach $1.3\cdot10^{-5}$ and $8.0\cdot10^{-6}$. At $19$~dB none of the three codes produced an error in $3\times10^6$ codewords, and a crossing below a BLER of about $3\cdot10^{-7}$ cannot be resolved by simulation. This behaviour agrees with Table~\ref{tab:spectrum}: the $5$ times larger minimum determinant of $\mathcal{C}_8^{\times2}$ is offset by its larger low-determinant multiplicity, and \eqref{eq:ub} places the crossing far below any simulable error rate.
\item The unitary variant is $0.5$--$0.7$~dB behind the punctured perfect code.
\item The sphere decoder visits $1.7$--$1.9$ times more nodes for $\mathcal{C}_8^{\times2}$ than for the punctured perfect code, because its generator is not orthogonal.
\end{itemize}

\begin{figure}[t]
\centering
\includegraphics[width=\columnwidth]{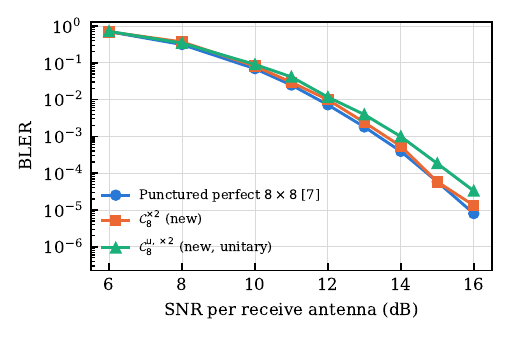}
\caption{BLER of rate-2 punctured $8\times8$ codes in an $8\times2$ system, $4$-QAM, i.i.d.\ Rayleigh, exact ML (sphere) decoding.}\label{fig:bler82}
\end{figure}

\subsection{Summary}
At the algebraic level, the new codes provide what the earlier transcendental codes could not: the NVD property, hence DMT optimality (Corollary~\ref{cor:dmt}), and provable determinant floors; in addition, $\mathcal{C}_8$ has the largest proved minimum determinant among the $8\times8$ codes compared. In simulation, $\mathcal{C}_4$ outperforms the $4\times4$ perfect code by about $0.4$--$0.5$~dB under exact ML decoding in all three channel models, and $\mathcal{C}_8^{\mathrm u}$ outperforms the $8\times8$ perfect code by $0.3$--$0.7$~dB at BLER $10^{-2}$--$10^{-3}$ under near-ML decoding, with the curves meeting near $5\cdot10^{-4}$ in i.i.d.\ fading. The transcendental codes fall behind at high SNR, by up to $0.43$~dB in correlated $4\times4$ fading. The asymptotic advantages of $\mathcal{C}_8$ and $\mathcal{C}_8^{\times2}$ are not visible down to BLER $2.5\cdot10^{-4}$ and $8\cdot10^{-6}$, respectively, exactly as the determinant-spectrum analysis predicts.

\section{Conclusion}\label{conclusion}
Non-cyclic crossed products can improve on cyclic perfect codes only if their structure constants are algebraic and the commuting-generator construction is abandoned. Theorems~\ref{thm:commuting} and \ref{thm:local} make this statement precise and identify the admissible Galois groups over $\Q(i)$, and Theorem~\ref{thm:centralizer} with Lemma~\ref{lem:tame} provides a practical division test. The twisted construction of Section~\ref{sec:construction} admits the explicit codeword form \eqref{eq:block}, a quaternion-like arrangement of cyclic codewords. With these tools we obtained an NVD, information-lossless $4\times4$ code and an infinite family of NVD, DMT-optimal $8\times8$ codes. The optimal member $\mathcal{C}_8$ has a minimum determinant at least $34.7$ times that of the $8\times8$ perfect code, and its rate-2 subcode has at least $5$ times that of the punctured perfect code. The determinant-spectrum analysis shows, and the simulations confirm, that these advantages are asymptotic. At practical error rates the low-determinant multiplicity governs performance: $\mathcal{C}_4$ outperforms the $4\times4$ perfect code by $0.4$--$0.5$~dB under exact ML decoding, $\mathcal{C}_8^{\mathrm u}$ outperforms the $8\times8$ perfect code by up to $0.7$~dB under near-ML decoding, and $\mathcal{C}_8^{\times2}$ remains within $0.25$~dB of the punctured perfect code down to BLER $10^{-5}$.

Several problems remain open:
\begin{itemize}
\item optimizing the multiplicity $A_4$ within these families through the choice of ideal, basis and $b$;
\item NVD $\Z/2^k\times\Z/2$ codes for $n=16$ (Corollary~\ref{cor:Z23} excludes $(\Z/2)^4$, but not $\Z/8\times\Z/2$);
\item fast-decodable orderings of the punctured codes, to be compared with \cite{SrinathRajan2014,VHO2012};
\item discriminant-optimal orders in the sense of \cite{HollantiLahtonenLu2008,VHLR2009}.
\end{itemize}

\section*{Reproducibility}
All generator matrices, the exact sphere decoder, a GPU-batched K-best decoder, and the scripts that produce every number and figure in this paper are available from the authors.

\end{document}